\newif\ifEntropy
\Entropytrue
\Entropyfalse

\ifEntropy
\documentclass[entropy,article,submit,moreauthors,pdftex,12pt,a4paper]{mdpi}
\lastpage{x}
\doinum{10.3390/------}
\pubvolume{xx}
\pubyear{2013}
\history{Received: xx / Accepted: xx / Published: xx}
\else
\documentclass[a4paper, 12pt, notitlepage]{scrartcl}
\fi
\usepackage[utf8]{inputenc}
\usepackage{amsmath}
\usepackage{amsthm}
\usepackage{amsfonts}
\usepackage{amssymb}

\usepackage[english]{babel}

\usepackage[x11names,rgb]{xcolor}
\usepackage{tikz}

\usetikzlibrary{arrows}
\usetikzlibrary{shapes}
\usetikzlibrary{backgrounds}

\usepackage{pgfplots}
  \pgfplotsset{articleplot/.style={
          axis x line=bottom,
          axis y line=left,
          axis line style={lightgray},
          y axis line style={lightgray},
          x axis line style={lightgray,-}
        }}

\usepackage{ifpdf}
\usepackage{url}

\usepackage{verbatim}

\usepackage{hyperref}

\newtheorem{lemma}{Lemma}
\newtheorem{prop}[lemma]{Proposition}
\newtheorem{thm}[lemma]{Theorem}
\newtheorem*{thm*}{Theorem}
\newtheorem{cor}[lemma]{Corollary}

\theoremstyle{definition}
\newtheorem{defn}[lemma]{Definition}
\theoremstyle{remark}
\newtheorem{rem}[lemma]{Remark}
\newtheorem{ex}[lemma]{Example}

\newcommand{\Rb}{\mathbb{R}}

\newcommand{\Acal}{\mathcal{A}}
\newcommand{\Scal}{\mathcal{S}}
\newcommand{\Xcal}{\mathcal{X}}
\newcommand{\Ycal}{\mathcal{Y}}
\newcommand{\Zcal}{\mathcal{Z}}

\newcommand{\uge}{\sqsupseteq}  

\newcommand{\Dpinn}{\Delta_{P}^{\ast}}
\newcommand{\TCI}{\widetilde{CI}}
\newcommand{\TSI}{\widetilde{SI}}
\newcommand{\TUI}{\widetilde{UI}}

\DeclareMathOperator*{\bigtimes}{\textnormal{\Large $\times$}} 
\DeclareMathOperator{\XOR}{\textnormal{XOR}}

\newcommand\Perp{\protect\mathpalette{\protect\independenT}{\perp}}
\def\independenT#1#2{\mathrel{\rlap{$#1#2$}\mkern2mu{#1#2}}}
\newcommand{\ind}[3][]{\left.#2 \Perp\!_{#1}\, #3 \inD}
\newcommand{\inD}[1][\relax]{\def\argone{#1}\def\temprelax{\relax}
  \ifx\argone\temprelax\right.\else\,\middle|#1\right.{}\fi}

\newcommand{\IHSP}{I_{\text{red}}}
\newcommand{\UIHSP}{UI_{\text{red}}}
\newcommand{\Imin}{I_{\min}}

\newcommand{\noPDF}[2]{\texorpdfstring{#1}{#2}}
\ifEntropy
\Title{Quantifying unique information}
\Author{Nils Bertschinger\noPDF{$^{1}$}{(1)}, Johannes Rauh\noPDF{$^{1}$}{(1)}*, Eckehard Olbrich\noPDF{$^{1}$}{(1)}, Jürgen Jost\noPDF{$^{1,2}$}{(1,2)}, Nihat Ay\noPDF{$^{1,2}$}{(1,2)}}
\address{\noPDF{$^{1}$}{(1) }Max Planck Institute for Mathematics in the Sciences, Inselstra\ss{}e 23, 04109 Leipzig, Germany\\
  \noPDF{$^{2}$}{(2) }Santa Fe Institute, 1399 Hyde Park Rd, Santa Fe, NM 87501, USA}
\corres{jrauh@mis.mpg.de}  
\abstract{ %
  We propose new measures of shared information, unique information and synergistic information that can be used to
  decompose the multi-information of a pair of random variables $(Y,Z)$ with a third random variable~$X$.  Our measures
  are motivated by an operational idea of unique information which suggests that shared information and unique
  information should depend only on the pair marginal distributions of $(X,Y)$ and~$(X,Z)$.  Although this invariance
  property has not been studied before, it is satisfied by other proposed measures of shared information.  The
  invariance property does not uniquely determine our new measures, but it implies that the functions that we define are
  bounds to any other measures satisfying the same invariance property.  We study properties of our measures and compare
  them to other candidate measures.}
\else
\title{Quantifying unique information}
\author{{\normalsize Nils Bertschinger$^{1}$, Johannes Rauh$^{1}$, Eckehard Olbrich$^{1}$, Jürgen Jost$^{1,2}$, Nihat Ay$^{1,2}$} \\
{\small $^{1}$Max Planck Institute for Mathematics in the Sciences, Leipzig, Germany} \\
{\small $^{2}$Santa Fe Institute, Santa Fe, USA}}
\date{October 22, 2013}
\fi

\begin{document}
\maketitle

\ifEntropy\else
\begin{abstract}
  We propose new measures of shared information, unique information and synergistic information that can be used to
  decompose the multi-information of a pair of random variables $(Y,Z)$ with a third random variable~$X$.  Our measures
  are motivated by an operational idea of unique information which suggests that shared information and unique
  information should depend only on the pair marginal distributions of $(X,Y)$ and~$(X,Z)$.  Although this invariance
  property has not been studied before, it is satisfied by other proposed measures of shared information.  The
  invariance property does not uniquely determine our new measures, but it implies that the functions that we define are
  bounds to any other measures satisfying the same invariance property.  We study properties of our measures and compare
  them to other candidate measures.

\end{abstract}
\fi

\section{Introduction}
\label{sec:introduction}

Consider three random variables $X,Y,Z$ with finite state spaces $\Xcal,\Ycal,\Zcal$.  Suppose that we are interested in
the value of~$X$, but we can only observe $Y$ or~$Z$.  If the tuple $(Y,Z)$ is not independent of~$X$, then the values
of $Y$ or $Z$ or both of them contain information about~$X$.  The information about $X$ contained in the tuple $(Y,Z)$
can be distributed in different ways.  For example, it may happen that $Y$ contains information about~$X$, but $Z$ does
not, or vice versa.  In this case, it would suffice to observe only one of the two variables $Y,Z$, namely the one
containing the information.  It may also happen, that $Y$ and $Z$ contain different information, so it would be
worthwhile to observe both of the variables.  If $Y$ and $Z$ contain the same information about~$X$, we could
choose to observe either $Y$ or~$Z$.  Finally, it is possible that neither $Y$ nor $Z$ taken for itself contains any
information about~$X$, but together they contain information about~$X$.  This effect is called \emph{synergy}, and it
occurs, for example, if all variables $X,Y,Z$ are binary, and $X = Y\XOR Z$.  In general, all effects may be present at
the same time.  That is, the information that $(Y,Z)$ has about~$X$ is a mixture of \emph{shared information}
$SI(X:Y;Z)$ (that is, information contained both in $Y$ and in $Z$), \emph{unique information} $UI(X:Y\setminus Z)$ and
$UI(X:Z\setminus Y)$ (that is, information that only one of $Y$ and $Z$ has) and \emph{synergistic} or
\emph{complementary information} $CI(X:Y;Z)$ (that is, information that can only be retrieved when considering $Y$ and
$Z$ together).  It is often assumed that these three types of information are everything there is, but one may ask, of
course, whether there are further types of information.

The total information that $(Y,Z)$ has about $X$ can be quantified by the mutual information $MI(X:(Y,Z))$.  Decomposing $MI(X:(Y,Z))$ into shared information, unique information and synergistic information leads to four terms, as
\begin{equation}
  \label{eq:MI-decomposition}
  MI(X:(Y,Z)) = SI(X:Y;Z) + UI(X:Y\setminus Z) + UI(X:Z\setminus Y) + CI(X:Y;Z).
\end{equation}
The interpretation of the four terms as informations demands that they
should all be positive. Furthermore, it suggests that the following
identities also hold:
\begin{equation}
  \label{eq:MI-decomposition-2}
  \begin{split}
    MI(X:Y) &= SI(X:Y;Z) + UI(X:Y\setminus Z), \\
    MI(X:Z) &= SI(X:Y;Z) + UI(X:Z\setminus Y).
  \end{split}
\end{equation}
In the following, when we talk about a \emph{binary information decomposition} we mean a set of three functions $SI$,
$UI$ and $CI$ that satisfy~\eqref{eq:MI-decomposition} and~\eqref{eq:MI-decomposition-2}.

Combining the three equalities in~\eqref{eq:MI-decomposition} and~\eqref{eq:MI-decomposition-2} and using the chain rule
of mutual information
\[
MI(X:(Y,Z))= MI(X:Y)+ MI(X:Z|Y)
\]
yields the identity
\begin{multline}
  \label{eq:CoI}
  CoI(X;Y;Z) := MI(X:Y) - MI(X:Y|Z) \\
  = MI(X:Y) + MI(X:Z) - MI(X:(Y,Z)) = SI(X:Y;Z) - CI(X:Y;Z),
\end{multline}
which identifies the \emph{co-information}\footnote{The co-information was originally called \emph{interaction
    information} in~\cite{McGill54:Multivariate_information_transmission}.} with the difference of shared information
and synergistic in\-for\-ma\-tion.  It has been known for a long time, that a positive co-information is a sign of redundancy,
while a negative co-information expresses synergy~\cite{McGill54:Multivariate_information_transmission}.  However,
although there have been many attempts, as of currently, there has been no fully satisfactory solution to separate the
redundant and synergistic contributions to the co-information, and also a fully satisfying definition of the function
$UI$ is still missing.
Observe that, since we have three equations relating the four functions~$SI(X:Y;Z)$, $UI(X:Y\setminus Z)$,
$UI(X:Z\setminus Y)$ and~$CI(X:Y;Z)$, it suffices to specify one of them to compute the others.  When defining a
solution for the unique information~$UI$, this leads to the consistency equation
\begin{equation}
  \label{eq:union-info}
  MI(X:Z) + UI(X:Y\setminus Z) = MI(X:Y) + UI(X:Z\setminus Y).
\end{equation}
The value of~\eqref{eq:union-info} can be interpreted as the \emph{union information}, that is, the union of the
informations contained in $Y$ and in $Z$ without the synergy.

The problem to separate the contributions of shared information and synergistic information to the co-information is
probably as old as the definition of co-information itself.  Nevertheless, the co-information has been widely used as a
measure of synergy in the neurosciences; see, for
example,~\cite{SchneidmanBialekBerry03:Synergy_and_redundancy_in_population_codes,LathamNirenberg05:Synergy_and_redundancy_revisited}
and references therein.  The first general attempt to construct a consistent information decomposition into terms
corresponding to different combinations of shared and synergistic information is due to Williams and
Beer~\cite{WilliamsBeer:Nonneg_Decomposition_of_Multiinformation}.  See also the references
in~\cite{WilliamsBeer:Nonneg_Decomposition_of_Multiinformation} for other approaches to study multivariate information.
While the general approach of~\cite{WilliamsBeer:Nonneg_Decomposition_of_Multiinformation} is intriguing, the proposed
measure of shared information $\Imin$ suffers from serious flaws, which prompted a series of other papers trying to
improve these
results~\cite{GriffithKoch13:Quantifying,HarderSalgePolani13:Bivariate_redundant_information,BROJ13:Shared_information}.

In our current contribution, we propose to define the unique information as follows: 
Let $\Delta$  
be the set of all joint distributions of $X$, $Y$ and~$Z$.  Define
\begin{multline*}
  \Delta_{P} = \Big\{ Q\in\Delta : Q(X=x,Y=y)=P(X=x,Y=y)\\
  \text{ and }Q(X=x,Z=z)=P(X=x,Z=z)\text{ for all }x\in\Xcal,y\in\Ycal,z\in\Zcal \Big\}
\end{multline*}
as the set of all joint distributions which have the same marginal distributions on the pairs $(X,Y)$ and $(X,Z)$.
Then we define
\begin{align*}
  \TUI(X:Y\setminus Z)
  & = \min_{Q\in\Delta_{P}} MI_{Q}(X:Y|Z), \\
  \intertext{%
    where $MI_{Q}(X:Y|Z)$ denotes the conditional multi-information of $X$ and $Y$ given~$Z$, computed with respect to
    the joint distribution~$Q$.  Equation~\eqref{eq:CoI} implies %
  }
  \TSI(X:Y;Z) &= \max_{Q\in\Delta_{P}} CoI_{Q}(X;Y;Z), \\
  \TCI(X:Y;Z) &= MI(X:(Y,Z)) - \min_{Q\in\Delta_{P}}MI_{Q}(X:(Y,Z)).
\end{align*}
In Section~\ref{sec:properties} we show that the four functions $\TUI$, $\TSI$ and $\TCI$ are non-negative, and we study
further properties.  In Appendix~\ref{sec:computing} we describe the set $\Delta_{P}$ in terms of a parametrization.

Our approach is motivated by the idea that unique and shared information should only depend on the marginal distribution
of the pairs $(X,Z)$ and~$(X,Y)$.  This idea can be explained from an operational interpretation of unique information:
Namely, if $Y$ has unique information about $X$ (with respect to~$Z$), then there must be some way to extract this
information.  More precisely, there must be a situation in which $Y$ can use this information to perform better at
predicting the outcome of~$X$.
We make this idea precise in Section~\ref{sec:motivation} and show how it naturally leads to the definition of the
functions $\TUI$, $\TSI$ and~$\TCI$, as defined above.  Section~\ref{sec:properties} contains basic properties of these
three functions.  In particular, Lemma~\ref{lem:nonneg} shows that all three functions are non-negative.
Corollary~\ref{cor:Blackwell} proves that the function $\TUI$ is consistent with the operational idea put forward in
Section~\ref{sec:motivation}.  In Section~\ref{sec:comparison-with-I_HSP} we compare our function with other proposed
information decompositions.  Some examples are studied in Section~\ref{sec:examples}.  Remaining open problems are
discussed in Section~\ref{sec:several-variables}.  The appendix contains some more technical aspects that help to
compute the functions $\TUI$, $\TSI$ and~$\TCI$.

\section{Operational interpretation}
\label{sec:motivation}

Our basic idea to characterize unique information is the following: If $Y$ has unique information about~$X$ with respect
to~$Z$, then there must be some way to extract this information.  That is, there must be a situation in which this
unique information is useful.  We formalize this idea in terms of decision problems as follows:

Let $X$, $Y$, $Z$ be three random variables, let $p$ be the marginal distribution of~$X$, and let
$\kappa\in[0,1]^{\Xcal\times\Ycal}$ and $\mu\in[0,1]^{\Xcal\times\Zcal}$ be (row) stochastic matrices describing the
conditional distribution of $Y$ and $Z$, respectively, given~$X$.  In other words, $p$, $\kappa$ and $\mu$ satisfy
\begin{equation*}
  P(X=x,Y=y) = p(x)\kappa(x;y)
  \quad\text{ and }\quad
  P(X=x,Z=z) = p(x)\mu(x;z).
\end{equation*}
Observe that, if $p(x)>0$, then $\kappa(x;y)$ and $\mu(x;z)$ are uniquely defined.  Otherwise, $\kappa(x;y)$ and
$\mu(x;z)$ can be chosen arbitrarily.  In this section, we will assume that the random variable $X$ has full support.
If this is not the case, our discussion will remain valid after replacing $\Xcal$ by the support of~$X$.  In fact, the
information quantities that we consider later will not depend on those matrix elements $\kappa(x;y)$ and $\mu(x;z)$
which are not uniquely defined.

Suppose that an agent has a finite set of possible actions~$\Acal$.  After the agent chooses her action~$a\in\Acal$, she
receives a reward $u(x,a)$, which not only depends on the chosen action $a\in\Acal$, but also on the value $x\in\Xcal$
of the random variable~$X$.  The tuple $(p,\Acal,u)$, consisting of the prior distribution~$p$, the set of possible
actions~$A$ and the reward function~$u$ is called a \emph{decision problem}.  If the agent can observe the value $x$ of
$X$ before choosing her action, her best strategy is to chose $a$ such that $u(x,a) = \max_{a'\in\Acal}u(x,a')$.  Suppose
now, that the agent cannot observe $X$ directly, but the agent knows the probability distribution $p$ of~$X$.  Moreover,
the agent observes a random variable $Y$ with conditional distribution described by the row-stochastic
matrix~$\kappa\in[0,1]^{\Xcal\times\Ycal}$.  In this context, $\kappa$ will also be called a \emph{channel} from $\Xcal$
to~$\Ycal$.
When using a channel~$\kappa$, the agent's optimal strategy is to choose her action such that her expected reward
\begin{equation}
  \label{eq:exp-reward}
  \sum_{x}P(X=x|Y=y)u(x,a) = \frac{\sum_{x}p(x)\kappa(x;y) u(x,a)}{\sum_{x\in\Xcal}p(x)\kappa(x;y)}
\end{equation}
is maximal.  Note that, in order to maximize~\eqref{eq:exp-reward}, the agent has to know (or estimate) the prior
distribution of~$X$ as well as the channel~$\kappa$.  Often, the agent is allowed to play a stochastic strategy.
However, in the present setting, the agent cannot increase her expected reward by randomizing her actions, and
therefore, we only consider deterministic strategies here.

Let $R(\kappa,p,u,y)$ be the maximum of~\eqref{eq:exp-reward} (over $a\in\Acal$), and let
\begin{equation*}
  R(\kappa,p,u) = \sum_{y}P(Y=y) R(\kappa,p,u,y).
\end{equation*}
be the maximal expected reward that the agent can achieve by always choosing the optimal action.

In this setting we make the following definition:
\begin{defn}
  \label{def:unq-info}
  Let $X,Y,Z$ be three random variables, and let $p$ be the marginal distribution of~$X$.
  \begin{itemize}
  \item%
    $Y$ has \emph{unique information} about $X$ (with respect to~$Z$), if there is a set $\Acal$ and a reward
    function $u\in\Rb^{\Xcal\times\Acal}$ such that $R(\kappa,p,u)> R(\mu,p,u)$.
  \item%
    $Z$ has \emph{no unique information} about $X$ (with respect to~$Y$), if for any set $\Acal$ and reward function
    $u\in\Rb^{\Xcal\times\Acal}$ the inequality $R(\kappa,p,u)\ge R(\mu,p,u)$ holds.  In this situation we also say that
    $Y$ knows everything that $Z$ knows about~$X$, and we write $Y\uge_{X} Z$.
  \end{itemize}
\end{defn}
This operational idea allows to distinguish when the unique information vanishes, but, unfortunately, does not allow to
quantify the unique information.

As shown recently in~\cite{BR13:Blackwell_relation_and_zonotopes}, the question whether~$Y\uge_{X}Z$ or not, does not
depend on the prior distribution~$p$ (but just on the support of~$p$, which we assume to be~$\Xcal$).
In fact, if $p$ has full support, then, in order to check whether~$Y\uge_{X}Z$, it suffices to know the stochastic
matrices~$\kappa,\mu$ representing the conditional distributions of~$Y$ and $Z$ given~$X$.

Consider the case $\Ycal=\Zcal$ and~$\kappa=\mu\in K(\Xcal;\Ycal)$, i.e. $Y$ and $Z$ use a similar channel. In this case, $Y$ has no unique information with
respect to~$Z$, and $Z$ has no unique information with respect to~$Y$.  Hence, in the
decomposition~\eqref{eq:MI-decomposition} only the shared information and the synergistic information may be larger than
zero.  The shared information may be computed from
\begin{equation*}
  SI(X:Y;Z) = MI(X:Y) - UI(X:Y\setminus Z) = MI(X:Y) = MI(X:Z);
\end{equation*}
and so the synergistic information is
\begin{equation*}
  CI(X:Y;Z) = MI(X:(Y,Z)) - SI(X:Y;Z) = MI(X:(Y,Z)) - MI(X:Y).
\end{equation*}
Observe that in this case, the shared information can be computed from the marginal distribution of $X$ and~$Y$.  Only
the synergistic information depends on the joint distribution of $X$, $Y$ and~$Z$.

We argue that this should be the case in general: By what was said above, whether the unique information
$UI(X:Y\setminus Z)$ is greater than zero only depends on the two channels~$\kappa$ and~$\mu$.  Even more is true: The
set of decision problems $(p,\Acal,u)$ such that $R(\kappa,p,u)> R(\mu,p,u)$ only depends on $\kappa$ and~$\mu$ (and the
support of~$p$).  To quantify the unique information, this set of decision problems must be measured in some way.  It is
reasonable to expect that this quantification can be achieved by taking into account only the marginal distribution $p$
of~$X$.  
Therefore, we believe that a sensible measure $UI$ for unique information should satisfy the following property:
\begin{equation}
  \tag{$\ast$}
  \label{eq:assumption}
  \text{$UI(X:Y\setminus Z)$ only depends on $p$, $\kappa$ and~$\mu$.}
\end{equation}
Although this condition seems to have not been considered before, many candidate measures of unique information satisfy
this property; for example those defined
in~\cite{WilliamsBeer:Nonneg_Decomposition_of_Multiinformation,HarderSalgePolani13:Bivariate_redundant_information}.
In the following, we explore the consequences of assumption~\eqref{eq:assumption}.
\begin{lemma}
  \label{lem:SI-depends}
  Under assumption~\eqref{eq:assumption}, the shared information only depends on $p$, $\kappa$ and~$\mu$.
\end{lemma}
\begin{proof}
  This follows from $SI(X:Y;Z) = MI(X:Y) - UI(X:Y\setminus Z)$.
\end{proof}

Let $\Delta$ be the set of all joint distributions of $X$, $Y$ and~$Z$.  Fix $P\in\Delta$, and assume that the marginal
distribution of $X$, denoted by $p$, has full support.  Denote by $\kappa$ and $\mu$ the stochastic matrices
corresponding to the conditional distributions of $Y$ and $Z$ given~$X$.  Let
\begin{multline*}
  \Delta_{P} = \Big\{ Q\in\Delta_{P} : Q(X=x,Y=y)=P(X=x,Y=y)\\
  \text{ and }Q(X=x,Z=z)=P(X=x,Z=z)\text{ for all }x\in\Xcal,y\in\Ycal,z\in\Zcal \Big\}
\end{multline*}
be the set of all joint distributions which have the same marginal distributions on the pairs $(X,Y)$ and $(X,Z)$, and let
\begin{equation*}
  \Dpinn=\big\{ Q\in\Delta_{P} : Q(x)>0\text{ for all }x\in\Xcal \big\}
\end{equation*}
be the subset of distributions with full support.  Lemma~\ref{lem:SI-depends} says that, under
assumption~\eqref{eq:assumption}, the functions $UI(X:Y\setminus Z)$, $UI(X:Z\setminus Y)$ and $SI(X:Y;Z)$ are constant
on~$\Dpinn$, and only the function $CI(X:Y;Z)$ depends on the joint distribution~$Q\in\Dpinn$.  If we further assume
continuity, the same statement holds true for all~$Q\in\Delta_{P}$.  To make clear that we now consider the synergistic
information and the mutual information as a function of the joint distribution~$Q\in\Delta$, we write $CI_{Q}(X:Y;Z)$
and $MI_{Q}(X:(Y,Z))$ in the following; and we omit this subscript, if these information theoretic quantities are
computed with respect to the ``true'' joint distribution~$P$.

Consider the following functions:
\begin{align*}
  \TUI(X:Y\setminus Z)  & = \min_{Q\in\Delta_{P}} MI_{Q}(X:Y|Z), \\
  \TUI(X:Z\setminus Y)  & = \min_{Q\in\Delta_{P}} MI_{Q}(X:Z|Y), \\
  \TSI(X:Y;Z) &= \max_{Q\in\Delta_{P}} CoI_{Q}(X;Y;Z), \\
  \TCI(X:Y;Z) &= MI(X:(Y,Z)) - \min_{Q\in\Delta_{P}}MI_{Q}(X:(Y,Z)).
\end{align*}
Observe that these minima and maxima are well-defined, since the set $\Delta_{P}$ is compact and the mutual informations
and the co-information are continuous.  The next lemma says that, under assumption~\eqref{eq:assumption}, the quantities
$\TUI$, $\TSI$ and $\TCI$ bound the unique, shared and synergistic information.
\begin{lemma}
  \label{lem:bounds}
  Let $UI(X:Y\setminus Z)$, $UI(X:Z\setminus Y)$, $SI(X:Y;Z)$ and $CI(X:Y;Z)$ be non-negative continuous functions on
  $\Delta$ satisfying equations~\eqref{eq:MI-decomposition} and~\eqref{eq:MI-decomposition-2} and
  assumption~\eqref{eq:assumption}.  Then
  \begin{align*}
    UI(X:Y\setminus Z)  & \le \TUI(X:Y\setminus Z),  \\
    UI(X:Z\setminus Y)  & \le \TUI(X:Z\setminus Y),  \\
    SI(X:Y;Z)           & \ge \TSI(X:Y;Z),           \\
    CI(X:Y;Z)           & \ge \TCI(X:Y;Z).
  \end{align*}
  If $P\in\Delta$ and if there exists $Q\in\Delta_{P}$ such that $CI_{Q}(X:Y;Z)=0$, then equality holds in all four
  inequalities.  Conversely, if equality holds in one of the inequalities for a joint distribution~$P\in\Delta$, then
  there exists $Q\in\Delta_{P}$ such that $CI_{Q}(X:Y;Z)=0$.
\end{lemma}

\begin{proof}
  Fix a joint distribution $P\in\Delta$.  By Lemma~\ref{lem:SI-depends}, assumption~\eqref{eq:assumption} and
  continuity, the functions $UI(X:Y\setminus Z)$, $UI(X:Z\setminus Y)$ and $SI(X:Y;Z)$ are constant on~$\Delta_{P}$, and
  only the function $CI(X:Y;Z)$ depends on the joint distribution~$Q\in\Delta_{P}$.  The
  decomposition~\eqref{eq:MI-decomposition} rewrites to
  \begin{equation*}
    CI_{Q}(X:Y;Z) = MI_{Q}(X:(Y,Z)) - UI(X:Y\setminus Z) - UI(X:Z\setminus Y) - SI(X:Y;Z).
  \end{equation*}
  Using the non-negativity of synergistic information, this implies
  \begin{equation*}
    UI(X:Y\setminus Z) + UI(X:Z\setminus Y) + SI(X:Y;Z) \le \min_{Q\in\Delta_{P}} MI_{Q}(X:(Y,Z)).
  \end{equation*}
  In total, this shows
  \begin{equation*}
    CI(X:Y;Z) \ge MI(X:(Y,Z)) - \min_{Q\in\Delta_{P}}MI_{Q}(X:(Y,Z)) = \TCI(X:Y;Z).
  \end{equation*}

  The chain rule of mutual information says
  \begin{equation*}
    MI_{Q}(X:(Y,Z)) = MI_{Q}(X:Z) + MI_{Q}(X:Y|Z).
  \end{equation*}
  Now, $Q\in\Delta_{P}$ implies $MI_{Q}(X:Z)=MI(X:Z)$, and therefore,
  \begin{equation*}
    \TCI(X:Y;Z) = MI(X:Y|Z) - \min_{Q\in\Delta_{P}}MI_{Q}(X:Y|Z).
  \end{equation*}
  Moreover,
  \begin{equation*}
    MI_{Q}(X:Y|Z) = H_{Q}(X|Z) - H_{Q}(X|Y,Z),
  \end{equation*}
  where $H_{Q}(X|Z) = H(X|Z)$ for $Q\in\Delta_{P}$, and so
  \begin{equation*}
    \TCI(X:Y;Z) = \max_{Q\in\Delta_{P}}H_{Q}(X|Y,Z) -H(X|Y,Z).
  \end{equation*}
  
  By~\eqref{eq:CoI}, the shared information satisfies
  \begin{align*}
    SI(X:Y;Z) & = CI(X:Y;Z) + MI(X:Y) + MI(X:Z) - MI(X:(Y,Z)) \\
    &\ge \TCI(X:Y;Z) + MI(X:Y) + MI(X:Z) - MI(X:(Y,Z))
    \\ & = MI(X:Y) + MI(X:Z) - \min_{Q\in\Delta_{P}} MI_{Q}(X:(Y,Z))
    \\ & = \max_{Q\in\Delta_{P}} (MI_{Q}(X:Y) + MI_{Q}(X:Z) - MI_{Q}(X:(Y,Z)))
    \\ & = \max_{Q\in\Delta_{P}} CoI_{Q}(X;Y;Z) = \TSI(X:Y;Z).
  \end{align*}
  By~\eqref{eq:MI-decomposition-2}, the unique information satisfies
  \begin{align*}
    UI(X:Y\setminus Z)
    & = MI(X:Y) - SI(X:Y;Z) \\
    & \le \min_{Q\in\Delta_{P}}(MI_{Q}(X:(Y,Z) - MI(X:Z)) \\
    & = \min_{Q\in\Delta_{P}}(MI_{Q}(X:Y|Z)) = \TUI(X:Y\setminus Z).
  \end{align*}
  The inequality for $UI(X:Z\setminus Y)$ follows similarly.

  If there exists $Q_{0}\in\Delta_{P}$ such that $CI_{Q_{0}}(X:Y;Z)=0$, then
  \begin{equation*}
    0 = CI_{Q_{0}}(X:Y;Z) \ge \TCI(X:Y;Z) = MI(X:(Y,Z)) - \min_{Q\in\Delta_{P}}MI_{Q}(X:(Y,Z)) \ge 0.
  \end{equation*}
  Hence, in this case, all inequalities are tight.  Conversely, assume that one of the inequalities is tight for
  some~$P\in\Delta$.  The proof above shows that all four inequalities hold with equality.  By
  assumption~\eqref{eq:assumption}, the functions $\TUI$ and $\TSI$ are constant on~$\Delta_{P}$.  Therefore, the
  inequalities are tight for all~$Q\in\Delta_{P}$.  Now, if $Q_{0}\in\Delta_{P}$ minimizes $MI_{Q}(X:(Y,Z))$ over
  $\Delta_{P}$, then $CI_{Q_{0}}(X:Y;Z) = \TCI_{Q_{0}}(X:Y;Z) = 0$.
\end{proof}

The proof of Lemma~\ref{lem:bounds} shows that the optimization problems defining $\TUI$, $\TSI$ and $\TCI$ are in fact
equivalent; that is, it suffices to solve one of them.  Lemma~\ref{lem:opt-prob} in Section~\ref{sec:properties} gives
yet another formulation and shows that the solution is actually unique.

In the following, we interpret the functions $\TUI$, $\TSI$ and
$\TCI$ as measures of unique, shared and complementary information.
Under assumption~\eqref{eq:assumption}, Lemma~\ref{lem:bounds} says
that choosing those measures is equivalent to saying that in each set
$\Delta_{P}$ there exists a measure $Q$ such that $CI_{Q}(X:Y;Z)=0$.
In other words, $\TUI$, $\TSI$ and $\TCI$ are the only measures of
unique, shared and complementary information that satisfy the following property:
\begin{multline}
  \tag{$\ast\ast$}
  \label{eq:assumption-two}
  \text{It is not possible
    to decide whether or not there is synergistic information,} \\
  \text{when only the marginal distributions of $(X,Y)$ and $(X,Z)$ are known.}
\end{multline}
For any other combination of measures different from $\TUI$, $\TSI$ and $\TCI$ that satisfy
assumption~\eqref{eq:assumption} there are combinations of $(p,\mu,\kappa)$ such that the existence of non-vanishing
complementary information can be deduced.
Since complementary information should capture precisely the information that is carried by the joint dependencies
between $X$, $Y$ and~$Z$ we find assumption~\eqref{eq:assumption-two} natural, and we consider this observation as
evidence in favour of our interpretation of the functions $\TUI$, $\TSI$ and~$\TCI$.

\section{Properties}
\label{sec:properties}

\subsection{Characterization and Positivity}
\label{sec:positivity}

The next lemma shows that the optimization problems involved in the definitions of $\TUI$, $\TSI$ and $\TCI$ are easy to
solve numerically, in the sense that they are convex optimization problems on convex sets.  As always, theory is easier
than practice, as discussed in Example~\ref{ex:illconditioned} in Appendix~\ref{sec:crit-eq}.
\begin{lemma}
  \label{lem:opt-prob}
  Let $P\in\Delta$ and $Q_{P}\in\Delta_{P}$.  The following conditions are equivalent:
  \begin{enumerate}
  \item $MI_{Q_{P}}(X:Y|Z) = \min_{Q\in\Delta_{P}}MI_{Q}(X:Y|Z)$.
  \item $MI_{Q_{P}}(X:Z|Y) = \min_{Q\in\Delta_{P}}MI_{Q}(X:Z|Y)$.
  \item $MI_{Q_{P}}(X;(Y,Z)) = \min_{Q\in\Delta_{P}}MI_{Q}(X:(Y,Z))$.
  \item $CoI_{Q_{P}}(X;Y;Z) = \max_{Q\in\Delta_{P}}CoI_{Q}(X;Y;Z)$.
  \item $H_{Q_{P}}(X|Y,Z) = \max_{Q\in\Delta_{P}}H_{Q}(X|Y,Z)$.
  \end{enumerate}
  Moreover, the functions $MI_{Q}(X:Y|Z)$, $MI_{Q}(X:Z|Y)$ and $MI_{Q}(X:(Y,Z))$ are convex on $\Delta_{P}$; and
  $CoI_{Q}(X;Y;Z)$ and $H_{Q}(X|Y,Z)$ are concave.  Therefore, for fixed $P\in\Delta$, the set of all
  $Q_{P}\in\Delta_{P}$ satisfying any of these conditions is convex.
\end{lemma}
\begin{proof}
  The conditional entropy $H_{Q}(X|Y,Z)$ is a concave function on~$\Delta$; therefore, the set of maxima is convex.  To
  show the equivalence of the five optimization problems and the convexity properties, it suffices to show that the
  difference of any two minimized functions and the sum of a minimized and a maximized function is constant
  on~$\Delta_{p}$.  Except for $H_{Q}(X|Y,Z)$ this follows from the proof of Lemma~\ref{lem:bounds}.  For
  $H_{Q}(X|Y,Z)$, this follows from the chain rule:
  \begin{align*}
    MI_{Q}(X:(Y,Z)) &= MI_{P}(X:Y) + MI_{Q}(X:Z|Y)\\
    &= MI_{P}(X:Y) + H_{P}(X|Y) - H_{Q}(X|Y,Z) = H(X) - H_{Q}(X|Y,Z). \qedhere
  \end{align*}
\end{proof}
The optimization problems mentioned in Lemma~\ref{lem:opt-prob} will be studied more closely in
Appendix~\ref{sec:computing}.

\begin{lemma}[Non-negativity]
  \label{lem:nonneg}
  $\TUI$, $\TSI$ and $\TCI$ are non-negative functions.
\end{lemma}
\begin{proof}
  $\TCI$ is non-negative by definition.  The functions $\TUI$ are non-negative, because they are obtained by minimizing
  mutual informations, which are non-negative.

  Consider the real function
  \begin{equation*}
    Q_{0}(X=x,Y=y,Z=z) =
    \begin{cases}
      \frac{P(X=x,Y=y)P(X=x,Z=z)}{P(X=x)}, & \text{ if }P(X=x)>0, \\
      0, & \text{ else.}
    \end{cases}
  \end{equation*}
  It is easy to check $Q_{0}\in\Delta_{P}$.  Moreover, with respect to~$Q_{0}$, the two random variables $Y$ and $Z$ are
  conditionally independent given~$X$, that is, $MI_{Q_{0}}(Y:Z|X)=0$.  But this implies
  \begin{equation*}
    CoI_{Q_{0}}(X;Y;Z) = MI_{Q_{0}}(Y:Z) - MI_{Q_{0}}(Y:Z|X) = MI_{Q_{0}}(Y:Z) \ge 0.
  \end{equation*}
  Therefore, $\TSI(X:Y;Z)= \max_{Q\in\Delta_{P}} CoI_{Q}(X;Y;Z) \ge CoI_{Q_{0}}(X;Y;Z) \ge 0$, showing that $\TSI$ is a
  non-negative function.
\end{proof}
In general, the measure $Q_{0}$ constructed in the proof of Lemma~\ref{lem:nonneg} does not satisfy the conditions of
Lemma~\ref{lem:opt-prob}.

\subsection{Vanishing shared and unique information}
\label{sec:vanishing}

In this section we study when $\TSI=0$ and when $\TUI=0$.  In particular, in Corollary~\ref{cor:Blackwell} we show that
$\TUI$ conforms with the operational idea put forward in Section~\ref{sec:motivation}.

\begin{lemma}
  \label{lem:no-unique}
  $\TUI(X:Y\setminus Z)$ vanishes if and only if there exists a row-stochastic matrix
  $\lambda\in[0,1]^{\Zcal\times\Ycal}$ such that
  \begin{equation*}
    P(X=x,Y=y) = \sum_{z\in\Zcal}P(X=x,Z=z)\lambda(z;y).
  \end{equation*}
\end{lemma}
\begin{proof}
  If $MI_{Q}(X:Y|Z) = 0$ for some~$Q\in\Delta_{P}$, then $X$ and $Y$ are independent given $Z$ with respect
  to~$Q$.  Therefore, there exists a stochastic matrix $\lambda\in[0,1]^{\Zcal\times\Ycal}$ satisfying
  \begin{multline*}
    P(X=x,Y=y) = Q(X=x,Y=y) = \sum_{z\in\Zcal}Q(X=x,Z=z)\lambda(z;y) \\
    = \sum_{z\in\Zcal}P(X=x,Z=z)\lambda(z;y).
  \end{multline*}
  Conversely, if such a matrix $\lambda$ exists, then the equality
  \begin{equation*}
    Q(X=x,Y=y,Z=z) = P(X=x,Z=z)\lambda(z;y)
  \end{equation*}
  defines a probability distribution $Q$ which lies in~$\Delta_{P}$.  Then
  \begin{equation*}
    \TUI(X:Y\setminus Z)\le MI_{Q}(X:Y|Z) = 0.\qedhere
  \end{equation*}
\end{proof}
The last result can be translated into the language of our motivational Section~\ref{sec:motivation} and says that
$\TUI$ is consistent with our operational idea of unique information:
\begin{cor}
  \label{cor:Blackwell}
  $\TUI(X:Z\setminus Y)=0$ if and only if $Z$ has no unique information about $X$ with respect to~$Y$ (according to
  Definition~\ref{def:unq-info}).
\end{cor}
\begin{proof}
  We need to show that decision problems can be solved with the channel $\kappa$ at least as well as with the channel
  $\mu$ if and only if $\mu=\kappa\lambda$ for some stochastic matrix~$\lambda$.  This result is known as Blackwell's
  theorem~\cite{Blackwell53:Equivalent_comparisons_of_experiments}; see
  also~\cite{BR13:Blackwell_relation_and_zonotopes}.
\end{proof}
\begin{cor}
  \label{cor:ident-dist}
  Suppose that $\Ycal=\Zcal$ and that the marginal distributions of the pairs $(X,Y)$ and $(X,Z)$ are identical.  Then
  \begin{align*}
    \TUI(X:Y\setminus Z) & = \TUI(X:Z\setminus Y)=0, \\
    \TSI(X:Y;Z) & = MI(X:Y) = MI(X:Z), \\
    \TCI(X:Y;Z) & = MI(X:Y|Z) = MI(X:Z|Y).
  \end{align*}
  In particular, under assumption~\eqref{eq:assumption}, there is no unique information in this situation.
\end{cor}
\begin{proof}
  Apply Lemma~\ref{lem:no-unique} with the identity matrix in the place of~$\lambda$.
\end{proof}

\begin{lemma}
  \label{lem:TSI-zero}
  $\TSI(X:Y;Z)=0$ if and only if $MI_{Q_{0}}(Y:Z)=0$, where $Q_{0}\in\Delta$ is the distribution constructed in the
  proof of Lemma~\ref{lem:nonneg}.
\end{lemma}
The proof of the lemma will be given in Appendix~\ref{sec:proofs}, since it relies on some
technical results from Appendix~\ref{sec:computing}, where $\Delta_{P}$ is characterized and the critical equations
corresponding to the optimization problems in Lemma~\ref{lem:opt-prob} are computed.
\begin{cor}
  If both $\ind{Y}{Z}[X]$ and~$\ind{Y}{Z}$, then $\TSI(X:Y;Z)=0$.
\end{cor}
\begin{proof}
  By assumption, $P=Q_{0}$.  Thus the statement follows from Lemma~\ref{lem:TSI-zero}.
\end{proof}

\subsection{The bivariate PI axioms}
\label{sec:PI-axioms}

In~\cite{WilliamsBeer:Nonneg_Decomposition_of_Multiinformation}, Williams and Beer proposed axioms that a measure of
shared information should satisfy.  We call these axioms the \emph{PI axioms} after the partial information
decomposition framework derived from these axioms in~\cite{WilliamsBeer:Nonneg_Decomposition_of_Multiinformation}.
In fact, the PI axioms apply to a measure of shared information that is defined for arbitrarily many random variables,
while our function $\TSI$ only measures the shared information of two random variables (about a third variable).
The PI axioms are as follows:
\begin{enumerate}
\item The shared information of $Y_{1},\dots,Y_{n}$ about $X$ is symmetric under permutations of $Y_{1},\dots,Y_{n}$.
  \hfill(symmetry)
\item The shared information of $Y_{1}$ about $X$ is equal to $MI(X:Y_{1})$.
  \hfill(self-redundancy)
\item The shared information of $Y_{1},\dots,Y_{n}$ about $X$ is less than the shared information of
  $Y_{1},\dots,Y_{n-1}$ about~$X$, with equality if $Y_{n-1}$ is a function of~$Y_{n}$.
  \hfill(monotonicity)
\end{enumerate}
Any measure $\TSI$ of bivariate shared information that is consistent with the PI axioms must obviously satisfy the
following two properties, which we call the \emph{bivariate PI axioms:}
\begin{enumerate}
\item[A)] $\TSI(X:Y;Z)=\TSI(X:Z;Y)$.
  \hfill(symmetry)
\item[B)] $\TSI(X:Y;Z)\le MI(X:Y)$, with equality if $Z$ is a function of~$Y$.

  \hfill(bivariate monotonicity)
\end{enumerate}
We do not claim that any function $\TSI$ that satisfies A) and B) can be extended to a measure of multivariate shared
information satisfying the PI axioms.  In fact, such a claim is false, and as discussed in
Section~\ref{sec:several-variables}, our bivariate function $\TSI$ is not extendable in this way.

The following two lemmas show that $\TSI$ satisfies the bivariate PI axioms, and they show cor\-re\-spond\-ing properties of
$\TUI$ and~$\TCI$.
\begin{lemma}[Symmetry]
  \label{lem:symmetry-consistency}
  \begin{align*}
    \TSI(X:Y;Z) &= \TSI(X:Z;Y), \\
    \TCI(X:Y;Z) &= \TCI(X:Z;Y), \\
    MI(X:Z) + \TUI(X:Y\setminus Z) & = MI(X:Y) + \TUI(X:Z\setminus Y).
  \end{align*}
\end{lemma}
\begin{proof}
  The first two equalities follow since the definitions of $\TSI$ and $\TCI$ are symmetric in~$Y$ and~$Z$.  The third
  equality follows from
  \begin{multline*}
    MI(X:Z) + \TUI(X:Y\setminus Z)
    = \min_{Q\in\Delta_{P}}(MI(X:Z) + MI_{Q}(X:Y|Z)) \\
    = \min_{Q\in\Delta_{P}}(MI(X:Y) + MI_{Q}(X:Z|Y))  = MI(X:Y) + \TUI(X:Z\setminus Y),
  \end{multline*}
  where the chain rule of mutual information was used.
\end{proof}
The third equality from Lemma~\ref{lem:symmetry-consistency} is the consistency condition~\eqref{eq:union-info}.

The following lemma is the inequality condition of the monotonicity axiom.
\begin{lemma}[Bounds]
  \label{lem:monotonicity}
  \begin{align*}
    \TSI(X:Y;Z) &\le MI(X:Y), \\
    \TCI(X:Y;Z) &\le MI(X:Y|Z), \\
    \TUI(X:Y\setminus Z) &\ge MI(X:Y) - MI(X:Z).
  \end{align*}  
\end{lemma}
\begin{proof}
  The first inequality follows from
  \begin{multline*}
    \TSI(X:Y;Z) = \max_{Q\in\Delta_{P}} CoI_{Q}(X;Y;Z) \\
     = \max_{Q\in\Delta_{P}} (MI(X:Y) - MI_{Q}(X:Y|Z))
    \le MI(X:Y),
  \end{multline*}
  the second from
  \begin{equation*}
    \TCI(X:Y;Z) = MI(X:Y|Z) - \min_{Q\in\Delta_{P}}MI_{Q}(X:Y|Z),
  \end{equation*}
  using the chain rule again.  The last inequality follows from the first inequality,
  equality~\eqref{eq:MI-decomposition-2} and the symmetry of Lemma~\ref{lem:symmetry-consistency}.
\end{proof}

To finish the study of the bivariate PI axioms, only the equality condition in the monotonicity axiom is missing.  We
show that $\TSI$ satisfies $\TSI(X:Y;Z)=MI(X:Y)$ not only if $Z$ is a deterministic function of~$Y$, but also more
generally, when $Z$ is independent of $X$ given~$Y$.  In this case, $Z$ can be interpreted as a stochastic function
of~$Y$, independent of~$X$.
\begin{lemma}
  \label{lem:X-indep-Z-given-Y}
  If $X$ is independent of $Z$ given~$Y$, then $P$ solves the optimization problems of Lemma~\ref{lem:opt-prob}.  In
  particular,
  \begin{align*}
    \TUI(X:Y\setminus Z)  & = MI(X:Y|Z), \\
    \TUI(X:Z\setminus Y)  & = 0, \\
    \TSI(X:Y;Z) &= MI(X;Z), \\
    \TCI(X:Y;Z) &= 0.
  \end{align*}
\end{lemma}
\begin{proof}
  If $X$ is independent of $Z$ given~$Y$, then
  \begin{equation*}
    MI(X:Z|Y) = 0 \le \min_{Q\in\Delta_{P}} MI_{Q}(X:Z|Y),
  \end{equation*}
  so $P$ minimizes $MI_{Q}(X:Z|Y)$ over~$\Delta_{P}$.
\end{proof}

\begin{rem}
  \label{rem:X-indep-Z-given-Y}
  In fact, Lemma~\ref{lem:X-indep-Z-given-Y} can be generalized as follows: In any binary information decomposition,
  equations~\eqref{eq:MI-decomposition} and~\eqref{eq:MI-decomposition-2} and the chain rule imply
  \begin{equation*}
    MI(X:Z|Y) = MI(X:(Y,Z)) - MI(X:Y)
    = UI(X:Z\setminus Y) + CI(X:Y;Z).
  \end{equation*}
  Therefore, if $MI(X:Z|Y)=0$, then $UI(X:Z\setminus Y)=0=CI(X:Y;Z)$.
\end{rem}

\subsection{Probability distributions with structure}
\label{sec:specstruct}

In this section we compute the values of $\TSI$, $\TCI$ and $\TUI$ for probability distributions with special structure.
If two of the variables are identical, then~$\TCI=0$ as a consequence of~Lemma~\ref{lem:X-indep-Z-given-Y} (see
Corollaries~\ref{cor:Y-is-Z} and~\ref{cor:X-is-Y}).  When $X=(Y,Z)$, then the same is true
(Proposition~\ref{prop:X-is-YZ}).  Moreover, in this case, $\TSI((Y,Z):Y;Z)=MI(Y:Z)$.  This equation has been postulated
as an additional axiom, called \emph{identity axiom}, in~\cite{HarderSalgePolani13:Bivariate_redundant_information}.

\begin{cor}
  \label{cor:Y-is-Z}
  \begin{align*}
    \TCI(X:Y;Y) & = 0, \\
    \TSI(X:Y;Y) & = CoI(X;Y;Y) = MI(X:Y), \\
    \TUI(X:Y\setminus Y) & = 0.
  \end{align*}  
\end{cor}
\begin{proof}
  If $Y=Z$, then $X$ is independent of $Z$ given~$Y$.
\end{proof}
\begin{cor}
  \label{cor:X-is-Y}
  \begin{align*}
    \TCI(X:X;Z) & = 0, \\
    \TSI(X:X;Z) & = CoI(X;X;Z) = MI(X:Z) - MI(X:Z|X) = MI(X:Z), \\
    \TUI(X:X\setminus Z) & = MI(X:X|Z) = H(X|Z), \\
    \TUI(X:Z\setminus X) & = MI(X:Z|X) = 0.
  \end{align*}  
\end{cor}
\begin{proof}
  If $X=Y$, then $X$ is independent of $Z$ given~$Y$.
\end{proof}
\begin{rem}
  \label{rem:Y-is-X-or-Z}
  Remark~\ref{rem:X-indep-Z-given-Y} implies that Corollaries~\ref{cor:Y-is-Z} and~\ref{cor:X-is-Y} hold for any
  bivariate information de\-com\-po\-si\-tion.
\end{rem}
\begin{prop}[Identity property]
  \label{prop:X-is-YZ}
  Suppose that $\Xcal=\Ycal\times\Zcal$, and $X=(Y,Z)$.  Then $P$ solves the optimization problems of
  Lemma~\ref{lem:opt-prob}.  In particular,
  \begin{align*}
    \TCI((Y,Z):Y;Z) & = 0, \\
    \TSI((Y,Z):Y;Z) & = MI(Y:Z), \\
    \TUI((Y,Z):Y\setminus Z) & = H(Y|Z), \\
    \TUI((Y,Z):Z\setminus Y) & = H(Z|Y).
  \end{align*}  
\end{prop}
\begin{proof}
  If $X=(Y,Z)$, then, by Corollary~\ref{cor:X-is-YZ:DP-empty} in the Appendix, $\Delta_{P}=\{P\}$, and therefore
  \begin{multline*}
    \TSI((Y,Z):Y;Z) =MI((Y,Z):Y) - MI((Y,Z):Y|Z) \\
    = H(Y) - H(Y|Z) = MI(Y:Z)
  \end{multline*}
  and
  \begin{equation*}
    \TUI((Y:Z):Y\setminus Z) = MI((Y,Z):Y|Z) = H(Y|Z),
  \end{equation*}
  and similarly for~$\TUI((Y:Z):Y\setminus Z)$.
\end{proof}

The following Lemma shows that the quantities $\TUI$, $\TSI$ and $\TCI$ are additive when considering systems that can
be decomposed into independent subsystems.
\begin{lemma}
  \label{lem:product-system}
  Let $X_{1},X_{2},Y_{1},Y_{2},Z_{1},Z_{2}$ be random variables such that $(X_{1},Y_{1},Z_{1})$ is independent of
  $(X_{2},Y_{2},Z_{2})$.  Then
  \begin{align*}
    \TSI((X_{1},X_{2}):(Y_{1},Y_{2});(Z_{1},Z_{2})) &= \TSI(X_{1}:Y_{1};Z_{1}) + \TSI(X_{1}:Y_{1};Z_{1}), \\
    \TCI((X_{1},X_{2}):(Y_{1},Y_{2});(Z_{1},Z_{2})) &= \TCI(X_{1}:Y_{1};Z_{1}) + \TCI(X_{1}:Y_{1};Z_{1}), \\
    \TUI((X_{1},X_{2}):(Y_{1},Y_{2})\setminus(Z_{1},Z_{2})) &= \TUI(X_{1}:Y_{1}\setminus Z_{1}) + \TUI(X_{1}:Y_{1}\setminus Z_{1}), \\
    \TUI((X_{1},X_{2}):(Z_{1},Z_{2})\setminus(Y_{1},Y_{2})) &= \TUI(X_{1}:Z_{1}\setminus Y_{1}) + \TUI(X_{1}:Z_{1}\setminus Y_{1}).
  \end{align*}
\end{lemma}
The proof of the last lemma is given in Appendix~\ref{sec:proofs}.

\section{Comparison with other measures}
\label{sec:comparison-with-I_HSP}

In this section we compare our information decomposition using $\TUI$, $\TSI$ and $\TCI$ with similar functions proposed
in other papers; in particular, the function $\Imin$ of~\cite{WilliamsBeer:Nonneg_Decomposition_of_Multiinformation} and
the bivariate redundancy measure $\IHSP$ of~\cite{HarderSalgePolani13:Bivariate_redundant_information}.  We do not
repeat their definitions here, since they are rather technical.

The first observation is that both $\IHSP$ and $\Imin$ satisfy assumption~\eqref{eq:assumption}.  Therefore,
$\IHSP\ge\TSI$ and $\Imin\ge\TSI$.  According to~\cite{HarderSalgePolani13:Bivariate_redundant_information}, $\Imin$
tends to be larger than~$\IHSP$, but there are some exceptions.

It is easy to find examples where $\Imin$ is unreasonably
large~\cite{HarderSalgePolani13:Bivariate_redundant_information,BROJ13:Shared_information}.  It is much more difficult
to distinguish $\IHSP$ and $\TSI$.  In fact, in many special cases the two measures $\IHSP$ and $\TSI$ agree, as
the following results show.
\begin{thm}
  \label{thm:IHSP-zero-iff-TSI-zero}
  $\IHSP(X:Y;Z)=0$ if and only if $\TSI(X:Y;Z)=0$.
\end{thm}
The proof of the theorem builds on the following lemma:
\begin{lemma}
  \label{lem:IHSP-zero}
  If both $\ind{Y}{Z}[X]$ and~$\ind{Y}{Z}$, then $\IHSP(X:Y;Z)=0$.
\end{lemma}
The proof of the lemma is deferred to Appendix~\ref{sec:proofs}.
\begin{proof}[Proof of Theorem~\ref{thm:IHSP-zero-iff-TSI-zero}]
  By Lemma~\ref{lem:bounds}, if $\IHSP(X:Y;Z)=0$, then $\TSI(X:Y;Z)=0$.  Now assume that~$\TSI(X:Y;Z)=0$.  Since both
  $\TSI$ and $\IHSP$ are constant on~$\Delta_{P}$, we may assume that $P=Q_{0}$; that is, we may assume that
  $\ind{Y}{Z}[X]$.  Then $\ind{Y}{Z}$ by Lemma~\ref{lem:TSI-zero}.  Therefore, Lemma~\ref{lem:IHSP-zero} implies that
  $\IHSP(X:Y;Z)=0$.
\end{proof}
Denote by $\UIHSP$ the unique information defined from $\IHSP$ and~\eqref{eq:MI-decomposition-2}.  Then:
\begin{thm}
  \label{thm:UIHSP-zero-TUI-zero}
  $\UIHSP(X:Y\setminus Z)=0$ if and only if $\TUI(X:Y\setminus Z)=0$
\end{thm}
\begin{proof}
  By Lemma~\ref{lem:bounds}, if $\TUI$ vanishes, then so does $\UIHSP$.  Conversely, $\UIHSP(X:Y\setminus Z)=0$ if
  and only if $\IHSP(X:Y;Z) = MI(X:Y)$.  By (20) in~\cite{HarderSalgePolani13:Bivariate_redundant_information}, this is
  equivalent to $p(x|y) = p_{y\searrow Z}(x)$ for all $x,y$.  In this case, $p(x|y) = \sum_{z}p(x|z)\lambda(z;y)$ for
  some $\lambda(z;y)$ with $\sum_{z}\lambda(z;y)=1$.
  Hence, Lemma~\ref{lem:no-unique} implies that $\TUI(X:Y\setminus Z)=0$.
\end{proof}
Theorem~\ref{thm:UIHSP-zero-TUI-zero} implies that $\IHSP$ does not contradict our operational ideas introduced in
Section~\ref{sec:motivation}.
\begin{cor}
  \label{cor:Ired-is-TSI-1}
  Suppose that one of the following conditions is satisfied:
  \begin{enumerate}
  \item $X$ is independent of $Y$ given~$Z$.
  \item $X$ is independent of $Z$ given~$Y$.
  \item $\Xcal=\Ycal\times\Zcal$, and $X=(Y,Z)$.
  \end{enumerate}
  Then $\IHSP(X:Y;Z) = \TSI(X:Y;Z)$.
\end{cor}
\begin{proof}
  If $X$ is independent of $Z$ given~$Y$, then, by Remark~\ref{rem:X-indep-Z-given-Y}, for any binary information
  decomposition, $UI(X:Z\setminus Y)=0$.  In particular, $\TUI(X:Z\setminus Y)=0=\UIHSP(X:Z\setminus Y)$ (compare also
  Lemma~\ref{lem:X-indep-Z-given-Y} and Theorem~\ref{thm:UIHSP-zero-TUI-zero}).  Therefore, $\IHSP(X:Y;Z) =
  \TSI(X:Y;Z)$.  If $\Xcal=\Ycal\times\Zcal$ and $X=(Y,Z)$, then $\TSI(X:Y;Z)=MI(Y:Z)=\IHSP(X:Y;Z)$ by
  Proposition~\ref{prop:X-is-YZ} and the identity axiom in~\cite{HarderSalgePolani13:Bivariate_redundant_information}.
\end{proof}
\begin{cor}
  \label{cor:Ired-is-TSI-2}
  If the two pairs $(X,Y)$ and $(X,Z)$ have the same marginal distribution, then $\IHSP(X:Y;Z) = \TSI(X:Y;Z)$.
\end{cor}
\begin{proof}
  In this case, $\TUI(X:Y;Z)=0=\UIHSP(X:Y;Z)$.
\end{proof}

Although $\TSI$ and $\IHSP$ often agree, they are different functions.  An example where $\TSI$ and $\IHSP$ have
different values is the dice example given at the end of the next section.  In particular, it follows that $\IHSP$ does
not satisfy property~\eqref{eq:assumption-two}.

\section{Examples}
\label{sec:examples}

Table~\ref{tab:examples} contains the values of $\TCI$ and $\TSI$ for some paradigmatic examples.  The list of examples
is taken from~\cite{HarderSalgePolani13:Bivariate_redundant_information}; see also~\cite{GriffithKoch13:Quantifying}.
In all these examples, $\TSI$ agrees with~$\IHSP$.
In particular, in these examples the values of $\TSI$ agree with the intuitively plausible values called ``expected
values'' in~\cite{HarderSalgePolani13:Bivariate_redundant_information}.


\begin{table}
  \centering
  \begin{tabular}{lcccl}
  Example     & $\TCI$ & $\TSI$ & $\Imin$ & Note \\
  \hline
  \textsc{Rdn}$^{\phantom{|^{|}}}$ & 0 & 1 & 1 & $X=Y=Z$ uniformly distributed
  \\
  \textsc{Unq} & 1 & 0 & 1 & $X=(Y,Z)$, $Y,Z$ i.i.d.
  \\
  \textsc{Xor} & 1 & 0 & 0 & $X=Y\operatorname{XOR}Z$, $Y,Z$ i.i.d.
  \\
  \textsc{And} & 1/2 & 0.311 & 0.311 & $X=Y\operatorname{AND}Z$, $Y,Z$ i.i.d.
  %
  \\
  \\
  \textsc{RdnXor} & 1 & 1 & 1 & $X=(Y_{1}\operatorname{XOR}Z_{1},W)$,\\
  & & & & $Y=(Y_{1},W)$, $Z=(Z_{1},W)$, $Y_{1},Z_{1},W$ i.i.d.
  \\
  \textsc{RdnUnqXor}$\!\!\!\!\!\!$ & 1 & 1 & 2 & $X=(Y_{1}\operatorname{XOR}Z_{1},(Y_{2},Z_{2}),W)$, $Y=(Y_{1},Y_{2},W)$, \\
  & & & & $Z=(Z_{1},Z_{2},W)$, $Y_{1},Y_{2},Z_{1},Z_{2},W$ i.i.d.
  \\
  \textsc{XorAnd} & 1 & 1/2 & 1/2 & $X=(Y\operatorname{XOR}Z,Y\operatorname{AND}Z)$, $Y,Z$ i.i.d. 
  \\ \\
  Copy & 0 & $\!\!\!MI(X:Y)\!\!\!\!\!\!\!$ & 1 & $X=(Y,Z)$
\end{tabular}

\caption{The value of $\TSI$ in some examples.  The note is a short explanation of the example; see~\cite{HarderSalgePolani13:Bivariate_redundant_information} for the details.}
\label{tab:examples}
\end{table}

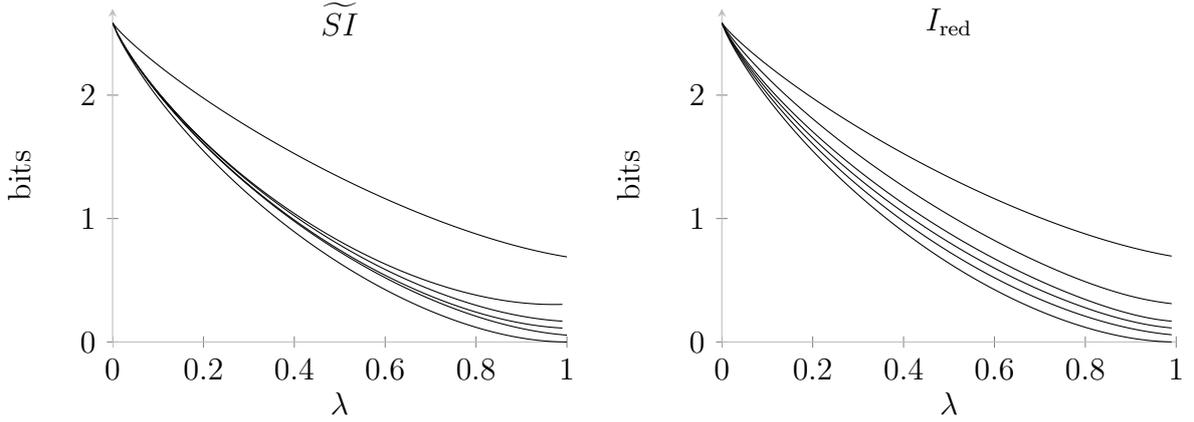
\begin{figure}
  \centering
  \begin{tikzpicture}
    \begin{axis}[articleplot,xlabel={$\lambda$}, ylabel={bits}, width=0.48\linewidth, height=6cm, xmin=0, xmax=1, ymin=0,
      ymax=2.7,title={$\TSI$},title style={yshift=-20pt}]  
      \addplot[black] table[x index=0, y index=1] {data/Dice1t.dat}; 
      \addplot[black] table[x index=0, y index=1] {data/Dice2t.dat}; 
      \addplot[black] table[x index=0, y index=1] {data/Dice3t.dat}; 
      \addplot[black] table[x index=0, y index=1] {data/Dice4t.dat}; 
      \addplot[black] table[x index=0, y index=1] {data/Dice5t.dat}; 
      \addplot[black] table[x index=0, y index=1] {data/Dice6t.dat}; 
    \end{axis}
  \end{tikzpicture}
  \hfill
  \begin{tikzpicture}
    \begin{axis}[articleplot,xlabel={$\lambda$}, ylabel={bits}, width=0.48\linewidth, height=6cm, xmin=0, xmax=1, ymin=0,
      ymax=2.7,title={$\IHSP$},title style={yshift=-20pt}]  
      \addplot[black] table[x index=0, y index=5] {data/dice1.dat}; 
      \addplot[black] table[x index=0, y index=5] {data/dice2.dat}; 
      \addplot[black] table[x index=0, y index=5] {data/dice3.dat}; 
      \addplot[black] table[x index=0, y index=5] {data/dice4.dat}; 
      \addplot[black] table[x index=0, y index=5] {data/dice5.dat}; 
      \addplot[black] table[x index=0, y index=5] {data/dice6.dat}; 
    \end{axis}
  \end{tikzpicture}
  \caption{The shared information measures $\TSI$ or~$\IHSP$ in the dice example depending on the correlation parameter
    $\lambda$ (figure on the right reproduced from~\cite{HarderSalgePolani13:Bivariate_redundant_information}).  The
    summation parameter $\alpha$ varies from 1 (uppermost line) to 6 (lowest line).}
\label{fig:dice}
\end{figure}
As a more complicated example we treated the following system with two parameters $\lambda\in[0,1]$,
$\alpha\in\{1,2,3,4,5,6\}$, also proposed by~\cite{HarderSalgePolani13:Bivariate_redundant_information}.  Let $Y$ and
$Z$ be two dice, and define $X=Y + \alpha Z$.  To change the degree of dependence of the two dice, assume that they are
distributed according to
\begin{equation*}
  P(Y=i,Z=j) = \frac{\lambda}{36} + (1-\lambda)\frac{\delta_{i,j}}{6}.
\end{equation*}
For $\lambda=0$ the two dice are completely correlated, while for $\lambda=1$ they are independent.  The resulting
shared information is shown in Figure~\ref{fig:dice}.  As a comparison, we reproduce Figure~8
from~\cite{HarderSalgePolani13:Bivariate_redundant_information} showing the function $\IHSP$ in the same example.  In
fact, for $\alpha=1$, $\alpha=5$ and $\alpha=6$ the two functions agree.  Moreover, they agree for $\lambda=0$ and
$\lambda=1$.  In all other cases, $\TSI\le\IHSP$, in agreement with Lemma~\ref{lem:bounds}.  For $\alpha=1$ and
$\alpha=6$ and $\lambda=0$ the fact that $\IHSP=\TSI$ follows from the results in
Section~\ref{sec:comparison-with-I_HSP}; in the other cases we do not know a simple reason for this coincidence.

It is interesting to note that for small $\lambda$ and $\alpha>1$ the function $\TSI$ depends only weakly on~$\alpha$.
In contrast, the dependence of $\IHSP$ on $\alpha$ is stronger.  At the moment we do not have an argument that tells us
which of these two behaviours is more intuitive.

\section{Outlook}
\label{sec:several-variables}

We defined a decomposition of the mutual information $MI(X:(Y,Z))$ of a random variable $X$ with a pair of random
variables~$(Y,Z)$ into non-negative terms which have an interpretation in terms of shared information, unique
information and synergistic information.  We have shown that the quantities $\TSI$, $\TCI$ and $\TUI$ have many
properties that such a decomposition should intuitively fulfil; among them the PI axioms and the identity axiom.  It is
a natural question whether the same can be done when further random variables are added to the system.

The first question in this context is how the decomposition of $MI(X:Y_{1},\dots,Y_{n})$ should look like.  How
many terms do we need?  In the bivariate case $n=2$, many people agree that shared, unique and synergistic information
should provide a complete decomposition (but it may well be worth to look for other types of decompositions).  For
$n>2$, there is no universal agreement of this kind.

Williams and Beer proposed a framework that suggests to construct an information decomposition only in terms of shared
information~\cite{WilliamsBeer:Nonneg_Decomposition_of_Multiinformation}.  Their ideas naturally lead to a decomposition
according to a lattice, called PI lattice.  For example, in this framework, $MI(X:Y_{1},Y_{2},Y_{3})$ has to be
decomposed into 18 terms with well-defined interpretation.  The approach is very appealing, since it is only based on
very natural properties of shared information (the PI axioms) and the idea that all information can be ``localized,'' in
the sense that, in an information decomposition, it suffices to classify information according to ``who knows what,''
that is, which information is shared by which subsystems.

Unfortunately, as shown in~\cite{BROJ13:Shared_information}, our function $\TSI$ cannot be generalized to the case~$n=3$
in the framework of the PI lattice.  The problem is that the identity axiom is incompatible with a non-negative
decomposition according to the PI lattice.

Even though we currently cannot extend our decomposition to~$n>2$, our bivariate decomposition can be useful for the
analysis of larger systems consisting of more than two parts.  For example, the quantity
\begin{equation*}
  \TUI(X:Y_{i}\setminus (Y_{1},\ldots,Y_{i-1},Y_{i+1},\ldots,Y_{n}))
\end{equation*}
can still be interpreted as the unique information of $Y_{i}$
about~$X$ with respect to all other variables, and it can be used to
assess the value of the $i$th variable, when synergistic contributions
can be ignored.  Furthermore the measure has the intuitive property
that the unique information cannot grow when additional variables are
taken into account:
\begin{lemma}
$\TUI(X:Y \setminus (Z_{1},\ldots,Z_{k})) \geq \TUI(X:Y \setminus (Z_{1},\ldots,Z_{k+1}))$.
\end{lemma}
\begin{proof}
  Let $P^{k}$ be the joint distribution of $X,Y,Z_{1},\dots,Z_{k}$, and let $P^{k+1}$ be the joint distribution of
  $X,Y,Z_{1},\dots,Z_{k},Z_{k+1}$.  By definition, $P^{k}$ is a marginal of~$P^{k+1}$.  For any $Q\in\Delta_{P^{k}}$,
  the distribution $Q'$ defined by
  \begin{equation*}
    Q'(x,y,z_{1},\dots,z_{k},z_{k+1}) :=
    \begin{cases}
      \frac{Q(x,y,z_{1},\dots,z_{k})P^{k+1}(x,z_{1},\dots,z_{k},z_{k+1})}{P^{k}(x,z_{1},\dots,z_{k})}, & \text{ if }P^{k}(x,z_{1},\dots,z_{k})>0, \\
      0, & \text{ else},
    \end{cases}
  \end{equation*}
  lies in~$\Delta_{P^{k+1}}$.  Moreover, $Q$ is the $(X,Y,Z_{1},\dots,Z_{k})$-marginal of~$Q'$, and $Z_{k+1}$ is
  independent of $Y$ given $X,Z_{1},\dots,Z_{k}$.  Therefore,
  \begin{align*}
    MI_{Q'}(X:Y|Z_{1},&\dots,Z_{k},Z_{k+1}) 
    \le MI_{Q'}(X,Z_{k+1}:Y|Z_{1},\dots,Z_{k})
    \\ &
    = MI_{Q'}(X:Y|Z_{1},\dots,Z_{k}) + MI_{Q'}(Z_{k+1}:Y|X,Z_{1},\dots,Z_{k})
    \\ &
    \le MI_{Q'}(X:Y|Z_{1},\dots,Z_{k}) = MI_{Q}(X:Y|Z_{1},\dots,Z_{k}).
  \end{align*}
  The statement now follows by taking the minimum over $Q\in\Delta_{P^{k}}$.
\end{proof}

Thus, we believe that our measure, which is well-motivated in
operational terms, can serve as a good starting point towards a
general decomposition of multi-variate information.




\appendix

\section*{Appendix: Computing $\TUI$, $\TSI$ and $\TCI$}  
\stepcounter{section}
\label{sec:computing}

\subsection[The optimization domain]{The optimization domain $\Delta_P$}
\label{sec:DeltaP}

By Lemma~\ref{lem:opt-prob}, to compute the quantities $\TUI$, $\TSI$ and $\TCI$, we need to solve a convex optimization
problem.  In this section we study some aspects of this problem.

First we describe $\Delta_{P}$.  For any set $\Scal$ let $\Delta(\Scal)$ be the set of probability distributions
on~$\Scal$, and let $A$ be the map $\Delta\to\Delta(\Xcal\times\Ycal)\times\Delta(\Xcal\times\Zcal)$
that takes a joint probability distribution of $X$, $Y$ and~$Z$ and computes the marginal distributions of the pairs
$(X,Y)$ and~$(X,Z)$.  Then $A$ is a linear map, and $\Delta_{P}=(P + \ker A)\cap\Delta$.  In particular, $\Delta_{P}$ is
the intersection of an affine space and a simplex; hence $\Delta_{P}$ is a polytope.

The matrix describing $A$ (and denoted by the same symbol in the following) is a well-studied object.  For example, $A$
describes the graphical model associated with the graph $Y$---$X$---$Z$.  The columns of $A$ define a polytope, called
\emph{marginal polytope}.  Moreover, the kernel of $A$ is known: Let
$\delta_{x,y,z}\in\Rb^{\Xcal\times\Ycal\times\Zcal}$ be the characteristic function of the point
$(x,y,z)\in\Xcal\times\Ycal\times\Zcal$, and let
\begin{equation*}
  \gamma_{x;y,y';z,z'} = \delta_{x,y,z} + \delta_{x,y',z'} - \delta_{x,y',z} - \delta_{x,y,z'}.
\end{equation*}
\begin{lemma}
  \label{lem:defect-A}
  The defect of~$A$ (that is, the dimension of~$\ker A$) is $|\Xcal|(|\Ycal|-1)(|\Zcal|-1)$.
  \begin{itemize}
  \item 
    The functions $\gamma_{x;y,y';z,z'}$ for all $x\in\Xcal$, $y,y'\in\Ycal$ and $z,z'\in\Zcal$ span $\ker A$.
  \item 
    For any fixed $y_{0}\in\Ycal$, $z_{0}\in\Zcal$, the functions $\gamma_{x;y_{0},y;z_{0},z}$ for all $x\in\Xcal$,
    $y\in\Ycal\setminus\{y_{0}\}$ and $z\in\Zcal\setminus\{z_{0}\}$ form a basis of $\ker A$.
  \end{itemize}
\end{lemma}
\begin{proof}
  See~\cite{HostenSullivant02:Reducible_and_Cyclic_Models}.
\end{proof}

The vectors $\gamma_{x;y,y';z,z'}$ for different values of~$x\in\Xcal$ have disjoint supports.  As the next lemma shows,
this can be used to write $\Delta_{P}$ as a Cartesian product of simpler polytopes.  Unfortunately, the function
$MI(X:(Y,Z))$ does not respect this product structures. In fact, the diagonal directions are important (see
Example~\ref{ex:illconditioned} below).
\begin{lemma}
  \label{lem:polytope-DeltaP}
  Let $P\in\Delta$.  For all $x\in\Xcal$ with $P(x)>0$ denote by
  \begin{multline*}
    \Delta_{P,x} = \Big\{ Q\in\Delta(\Ycal\times\Zcal) : Q(Y=y)=P(Y=y|X=x)\\
    \text{ and }Q(Z=z)=P(Z=z|X=x)\Big\}
  \end{multline*}
  the set of joint distributions of $Y$ and $Z$ such that the marginal distributions of $Y$ and $Z$ agree with the
  conditional distributions of $Y$ and $Z$ given~$X=x$.
  Then the map $\pi_{P}:\Delta_{P}\mapsto\bigtimes_{x\in\Xcal:P(x)>0}\Delta_{P,x}$ that maps each $Q\in\Delta_{P}$ to
  the family $(Q(\cdot|X=x))_{x\in\Xcal:P(x)>0}$ of conditional distributions of $Y$ and $Z$ given~$X=x$ for those
  $x\in\Xcal$ with~$P(X=x)>0$ is a linear
  bijection. 
\end{lemma}
\begin{proof}
  The image of $\pi_{P}$ is contained in $\bigtimes_{x\in\Xcal:P(x)>0}\Delta_{P,x}$ by definition of~$\Delta_{P}$.  The
  relation
  \begin{equation*}
    Q(X=x,Y=y,Z=z) = P(X=x)Q(Y=y,Z=z|X=x)
  \end{equation*}
  shows that $\pi_{P}$ is injective and surjective.  Since $\pi_{P}$ is in fact a linear map, the domain and the
  codomain of $\pi_{P}$ are affinely equivalent.
\end{proof}
Each Cartesian factor $\Delta_{P,x}$ of $\Delta_{P}$ is a fibre polytope of the independence model.
\begin{cor}
  \label{cor:X-is-YZ:DP-empty}
  If $X=(Y,Z)$, then $\Delta_{P}=\{P\}$.
\end{cor}
\begin{proof}
  By assumption, both conditional probability distributions $P(Y|X=x)$ and $P(Z|X=x)$ are point measures.  Therefore,
  each factor $\Delta_{P,x}$ consists of a single point; namely the conditional distribution $P(Y,Z|X=x)$ of $Y$ and $Z$
  given~$X$.  Hence, $\Delta_{P}$ is a singleton.
\end{proof}

\subsection{The critical equations}
\label{sec:crit-eq}

\begin{lemma}
  \label{lem:crit-eq}
  The derivative of $MI_{Q}(X:(Y,Z))$ in the direction $\gamma_{x;y,y';z,z'}$ is
  \begin{equation*}
    \log\frac{Q(x,y,z)Q(x,y',z')}{Q(x,y',z)Q(x,y,z')}\frac{Q(y',z)Q(y,z')}{Q(y,z)Q(y',z')}.
  \end{equation*}
  Therefore, $Q$ solves the optimization problems of Lemma~\ref{lem:opt-prob} if and only if
  \begin{equation}
    \label{eq:crit-ineq}
    \log\frac{Q(x,y,z)Q(x,y',z')}{Q(x,y',z)Q(x,y,z')}\frac{Q(y',z)Q(y,z')}{Q(y,z)Q(y',z')} \ge 0
  \end{equation}
  for all $x,y,y',z,z'$ with $Q + \epsilon\gamma_{x;y,y';z,z'}\in\Delta_{P}$ for $\epsilon>0$ small enough.
\end{lemma}
\begin{proof}
  The proof is by direct computation.
\end{proof}

\begin{ex}[The AND-example]
  Consider the binary case $\Xcal=\Ycal=\Zcal=\{0,1\}$, assume that $Y$ and $Z$ are independent and uniformly
  distributed, and suppose that $X=Y\operatorname{AND}Z$.  The underlying distribution $P$ is uniformly distributed on
  the four states $\{000,001,010,111\}$.  In this case, $\Delta_{P,1}=\{\delta_{Y=1,Z=1}\}$ is a singleton, and
  $\Delta_{P,0}$ consists of all probability distributions $Q$ of the form
  \begin{equation*}
    Q(Y=y,Z=z) =
    \begin{cases}
      \frac13 + \alpha', & \text{ if }(y,z) = (0,0), \\
      \frac13 - \alpha', & \text{ if }(y,z) = (0,1), \\
      \frac13 - \alpha', & \text{ if }(y,z) = (1,0), \\
         \alpha', & \text{ if }(y,z) = (1,1),
    \end{cases}
  \end{equation*}
  for some~$0\le\alpha'\le\frac13$.  Therefore, $\Delta_{P}$ is a one-dimensional polytope consisting of all probability
  distributions of the form
  \begin{equation*}
    Q_{\alpha}(Y=y,Z=z) =
    \begin{cases}
      \frac14 + \alpha, & \text{ if }(x,y,z) = (0,0,0), \\
      \frac14 - \alpha, & \text{ if }(x,y,z) = (0,0,1), \\
      \frac14 - \alpha, & \text{ if }(x,y,z) = (0,1,0), \\
         \alpha,        & \text{ if }(x,y,z) = (0,1,1), \\
         \frac14,       & \text{ if }(x,y,z) = (1,1,1), \\
         0,             & \text{ else,}
    \end{cases}
  \end{equation*}
  for some~$0\le\alpha\le\frac14$.  To compute the minimum of $MI_{Q_{\alpha}}(X:(Y,Z))$ over $\Delta_{P}$, we compute
  the derivative with respect to~$\alpha$ (which equals the directional derivative of $MI_{Q}(X:(Y,Z))$ in the direction
  $\gamma_{0;0,1;0,1}$ at $Q_{\alpha}$) and obtain:
  \begin{equation*}
    \log\frac{(\frac14 + \alpha)\alpha}{(\frac14 - \alpha)^{2}}\frac{(\frac14 - \alpha)^{2}}{(\frac14 + \alpha)^{2}}
    = \log\frac{\alpha}{\frac14 + \alpha}.
  \end{equation*}
  Since $\frac{\alpha}{\frac14 + \alpha}<1$ for all~$\alpha>0$, the function $MI_{Q_{\alpha}}(X:(Y,Z))$ has a unique
  minimum at~$\alpha=\frac14$.  Therefore,
  \begin{align*}
    \TUI(X:Y\setminus Z)  & = MI_{Q_{1/4}}(X:Y|Z) = 0 = \TUI(X:Z\setminus Y), \\
    \TSI(X:Y;Z) &= CoI_{Q_{1/4}}(X;Y;Z)
                 = MI_{Q_{1/4}}(X:Y)
                 = \frac34\log\frac43, \\
    \TCI(X:Y;Z) &= MI(X:(Y,Z)) - MI_{Q_{1/4}}(X:(Y,Z))
                 = \frac12\log2.
  \end{align*}
  In other words, in the AND-example there is no unique information, but only shared and synergistic information.  This
  follows, of course, also from Corollary~\ref{cor:ident-dist}.
\end{ex}
\begin{ex}
  \label{ex:illconditioned}
  The optimization problems in Lemma~\ref{lem:opt-prob} can be very ill-conditioned, in the sense that there are
  directions in which the function varies fast, and other directions in which the function varies slowly.  As an
  example, consider the example where $P$ is the distribution of three i.i.d.~uniform binary random variables.  In this
  case, $\Delta_{P}$ is a square.  Figure~\ref{fig:illcond} contains a heat map of the function~$CoI_{Q}$
  on~$\Delta_{P}$, where $\Delta_{P}$ is parametrized by
  \begin{equation*}
    Q(a,b) = P + a \gamma_{0;0,1;0,1} + b \gamma_{1;0,1;0,1},
    \qquad
    -\frac18 \le a \le \frac18,
    -\frac18 \le b \le \frac18.
  \end{equation*}
  Clearly, the function varies very little along one of the diagonals.  In fact, along this diagonal, $X$ is independent
  of $(Y,Z)$, corresponding to a very low synergy.

  Although in this case the optimising probability distribution $Q_{P}$
  is unique, it can be difficult to find.  For example, Mathematica's function \verb|FindMinimum| does not always find
  the true optimum out of the box (apparently, \verb|FindMinimum| cannot make use of the convex structure in the
  presence of constraints)~\cite{Mathematica80}.

  \begin{figure}
    \centering
    \begin{tikzpicture}
      \begin{axis}[xlabel={$a$}, ylabel={$b$}, ylabel style={rotate=-90}, width=0.6\linewidth, height=0.6\linewidth,
        xmin=-0.127, xmax=0.127, ymin=-0.127, ymax=0.127, 
        xticklabel style={/pgf/number format/.cd,fixed,precision=5}, yticklabel style={/pgf/number format/.cd,fixed,precision=5}]
        \addplot graphics [xmin=-0.125,xmax=0.125,ymin=-0.125,ymax=0.125] {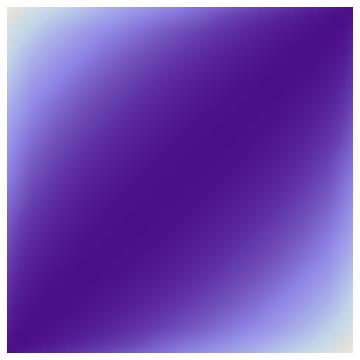};
      \end{axis}
      \ifEntropy  
      \node at (11,4.4) {\pgftext{\includegraphics{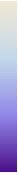}}};
      \node at (11,7.8) {large $\TCI$, low $CoI$};
      \node at (11,1.0) {low $\TCI$, large $CoI$};
      \else
      \node at (10,3.8) {\pgftext{\includegraphics{data/Palette}}};
      \node at (10,7.2) {large $\TCI$, low $CoI$};
      \node at (10,0.4) {low $\TCI$, large $CoI$};
      \fi
    \end{tikzpicture}
    \caption{The function $CoI_Q$ in Example~\ref{ex:illconditioned} (figure created with the help of
      Mathematica~\cite{Mathematica80}).  Darker colours indicate larger values of~$CoI_{Q}$.  In this example,
      $\Delta_{P}$ is a square.  The uniform distribution lies at the centre of this square and is the maximum
      of~$CoI_{Q}$.  In the two dark corners, $X$ is independent of $Y$ and $Z$, and either $Y=Z$ or $Y=\neg Z$.  In
      the two light corners, $Y$ and $Z$ are independent, and either $X=Y\XOR Z$ or $X=\neg(Y\XOR Z)$.}
    \label{fig:illcond}
  \end{figure}
\end{ex}

\subsection{Technical proofs}
\label{sec:proofs}

\begin{proof}[Proof of Lemma~\ref{lem:TSI-zero}]
  Since $\TSI(X:Y;Z)\ge CoI_{Q_{0}}(X;Y;Z)\ge 0$, if $\TSI(X:Y;Z)=0$, then
  \begin{equation*}
    0 = CoI_{Q_{0}}(X;Y;Z) = MI_{Q_{0}}(Y:Z) - MI_{Q_{0}}(Y:Z|X) = MI_{Q_{0}}(Y:Z).
  \end{equation*}
  To show that $MI_{Q_{0}}(Y:Z)=0$ is also sufficient, observe that
  \begin{equation*}
    Q_{0}(x,y,z)Q_{0}(x,y',z') = Q_{0}(x,y,z')Q_{0}(x,y',z),
  \end{equation*}
  by construction of $Q_{0}$, and that
  \begin{equation*}
    Q_{0}(y,z)Q_{0}(y',z') = Q_{0}(y,z')Q_{0}(y',z)
  \end{equation*}
  by the assumption that $MI_{Q_{0}}(Y:Z)=0$.  Therefore, by Lemma~\ref{lem:crit-eq}, all partial derivatives vanish
  at~$Q_{0}$.  Therefore, $Q_{0}$ solves the optimization problems in Lemma~\ref{lem:opt-prob}, and
  $\TSI(X:Y;Z)=CoI_{Q_{0}}(X;Y;Z)=0$.
\end{proof}

\begin{proof}[Proof of Lemma~\ref{lem:product-system}]
  Let $Q_{1}$ and $Q_{2}$ be solutions of the optimization problems from Lemma~\ref{lem:opt-prob} for $(X_{1},Y_{1},Z_{1})$
  and $(X_{2},Y_{2},Z_{2})$ in the place of $(X,Y,Z)$, respectively.  Consider the probability distribution~$Q$ defined
  by
  \begin{equation*}
    Q(x_{1},x_{2},y_{1},y_{2},z_{1},z_{2}) = Q_{1}(x_{1},y_{1},z_{1})Q_{2}(x_{2},y_{2},z_{2}).
  \end{equation*}
  Since $(X_{1},Y_{1},Z_{1})$ is independent of $(X_{2},Y_{2},Z_{2})$ (under~$P$), $Q\in\Delta_{P}$.  We show that $Q$
  solves the optimization problems from Lemma~\ref{lem:opt-prob} for
  $X=(X_{1},X_{2})$, $Y=(Y_{1},Y_{2})$ and $Z=(Z_{1},Z_{2})$.
  We use the notation from Appendix~\ref{sec:computing}.

  If $Q + \epsilon\gamma_{x_{1}x_{2};y_{1}y_{2},y'_{1}y'_{2};z_{1}z_{2},z'_{1}z'_{2}}\in\Delta_{Q}$, then
  \begin{equation*}
    Q_{1} + \epsilon\gamma_{x_{1};y_{1},y'_{1};z_{1},z'_{1}}\in\Delta_{Q_{1}}
    \quad\text{ and }\quad
    Q_{2} + \epsilon\gamma_{x_{2};y_{2},y'_{2};z_{2},z'_{2}}\in\Delta_{Q_{2}}.
  \end{equation*}
  Therefore, by Lemma~\ref{lem:crit-eq},
  \begin{align*}
    \log&\frac{Q(x_1x_2,y_1y_2,z_1z_2)Q(x_1x_2,y'_1y'_2,z'_1z'_2)}{Q(x_1x_2,y'_1y'_2,z_1z_2)Q(x_1x_2,y_1y_2,z'_1z'_2)}\frac{Q(y'_1y'_2,z_1z_2)Q(y_1y_2,z'_1z'_2)}{Q(y_1y_2,z_1z_2)Q(y'_1y'_2,z'_1z'_2)}
    \\ &
    = \log\frac{Q(x_1,y_1,z_1)Q(x_1,y'_1,z'_1)}{Q(x_1,y'_1,z_1)Q(x_1,y_1,z'_1)}\frac{Q(y'_1,z_1)Q(y_1,z'_1)}{Q(y_1,z_1)Q(y'_1,z'_1)}
    \\ & \quad
    + \log\frac{Q(x_2,y_2,z_2)Q(x_2,y'_2,z'_2)}{Q(x_2,y'_2,z_2)Q(x_2,y_2,z'_2)}\frac{Q(y'_2,z_2)Q(y_2,z'_2)}{Q(y_2,z_2)Q(y'_2,z'_2)}
    \ge 0,
  \end{align*}
  and hence, again by Lemma~\ref{lem:crit-eq}, $Q$ is a critical point and solves the optimization problems.
\end{proof}

\begin{proof}[Proof of Lemma~\ref{lem:IHSP-zero}]
  We use the notation from~\cite{HarderSalgePolani13:Bivariate_redundant_information}.
  The information divergence is jointly convex.  Therefore, any critical point of the divergence restricted to a convex
  set is a global minimizer.  Let $y\in\Ycal$.  Then it suffices to show: If $P$ satisfies the two conditional
  independence statements, then the marginal distribution $P_{X}$ of $X$ is a critical point of $D(P(\cdot|y)\|Q)$ for
  $Q$ restricted to $C_{\text{cl}}(\langle Z\rangle_{X})$; for if this statement is true, then $P_{y\searrow Z}=P_{X}$, thus
  $I^{\pi}_{X}(Y\searrow Z)=0$, and finally~$\IHSP(X:Y;Z)=0$.

  Let $z,z'\in\Zcal$.  The derivative of $D(P(\cdot|y)\|Q)$ at $Q=P_{X}$ in the direction $P(X|z) - P(X|z')$ is
  \begin{equation*}
    \sum_{x\in\Xcal}(P(x|z) - P(x|z'))\frac{P(x|y)}{P(x)}
    = \sum_{x\in\Xcal}\left(\frac{P(x,z)P(x,y)}{P(x)P(y)P(z)} - \frac{P(x,z')P(x,y)}{P(x)P(y)P(z')}\right).
  \end{equation*}
  Now, $\ind{Y}{Z}[X]$ implies
  \begin{equation*}
    \sum_{x\in\Xcal}\frac{P(x,z)P(x,y)}{P(x)}=P(y,z)
    \quad\text{ and }\quad
    \sum_{x\in\Xcal}\frac{P(x,z')P(x,y)}{P(x)}=P(y,z'),
  \end{equation*}
  and $\ind{Y}{Z}$ implies $P(y)P(z)=P(y,z)$ and $P(y)P(z')=P(y,z')$.  Together, this shows that $P_{X}$ is a critical
  point.
\end{proof}

\ifEntropy
\acknowledgements{Acknowledgements}
\else
\subsection*{Acknowledgements}
\fi

NB is supported by the Klaus Tschira Stiftung. JR acknowledges support from the VW Stiftung. EO has received funding
from the European Community's Seventh Framework Programme (FP7/2007-2013) under grant agreement no.~258749 (CEEDS) and
no.~318723 (MatheMACS). We thank Christoph Salge and Daniel Polani for fruitful discussions, and we thank Ryan James for
helpful comments on the manuscript.

\ifEntropy
\conflictofinterests{Conflicts of Interest}
The authors declare no conflicts of interest.
\fi

\ifEntropy
\bibliographystyle{mdpi}
\else
\bibliographystyle{plain}
\fi
\bibliography{Blackwell}

\end{document}